\documentclass[12pt]{article}

\usepackage{tocloft}
\usepackage{graphicx}
\usepackage{natbib}
\usepackage[protrusion=true,expansion=true]{microtype}
\usepackage{amsmath,amsthm,amsfonts,amssymb,mathrsfs,dsfont,mathtools}
\usepackage{booktabs}
\usepackage{url}
\usepackage{float}
\usepackage{color}
\usepackage{parskip}
\usepackage[explicit]{titlesec}
\usepackage{threeparttable}
\usepackage{setspace}
\usepackage[margin=0.85in]{geometry}
\usepackage[font = onehalfspacing]{caption}
\usepackage[usenames,dvipsnames]{xcolor}
\usepackage[hyperfootnotes=true]{hyperref}
\usepackage[labelfont=bf]{caption}
\usepackage{authblk}
\usepackage{caption}
\usepackage[titletoc,toc,title]{appendix}
\usepackage{enumerate}
\usepackage{soul}
\usepackage{subfig}
\usepackage{accents}
\usepackage{subfiles} 
\usepackage[capitalise, nameinlink]{cleveref}
\usepackage{dsfont}
\usepackage{enumitem}

\usepackage{changepage}
\usepackage[noend]{algpseudocode}
\usepackage{tocloft}
\usepackage[protrusion=true,expansion=true]{microtype}
\usepackage{amsmath,amsthm,amsfonts,amssymb,mathrsfs,dsfont,mathtools}
\usepackage{float}
\usepackage{color}
\usepackage{parskip}
\usepackage[explicit]{titlesec}
\usepackage{threeparttable}
\usepackage{setspace}
\usepackage[font = onehalfspacing]{caption}
\usepackage[dvipsnames]{xcolor}
\usepackage[hyperfootnotes=true]{hyperref}
\usepackage[labelfont=bf]{caption}
\usepackage{authblk}
\usepackage{enumerate}
\usepackage{soul}
\usepackage{subfig}
\usepackage{accents}
\usepackage[capitalise, nameinlink]{cleveref}
\usepackage{diagbox}
\usepackage{array,multirow}

\usepackage[utf8]{inputenc} % allow utf-8 input
\usepackage[T1]{fontenc}    % use 8-bit T1 fonts
\usepackage{hyperref}       % hyperlinks
\usepackage{url}            % simple URL typesetting
\usepackage{booktabs}       % professional-quality tables
\usepackage{amsfonts}       % blackboard math symbols
\usepackage{nicefrac}       % compact symbols for 1/2, etc.
\usepackage{microtype}      % microtypography
\usepackage{lipsum}		% Can be removed after putting your text content 
\usepackage{graphicx}
\usepackage{natbib}
\usepackage{doi}
\usepackage{subfiles}

\usepackage{adjustbox}
\usepackage{multirow}
\usepackage{amsmath}
\usepackage[section]{placeins}
\usepackage{algorithm}
\usepackage{algpseudocode}
\usepackage{dsfont}
\usepackage{comment}

\DeclareUnicodeCharacter{2212}{-}

\setstcolor{red} 

\numberwithin{equation}{section}

\hypersetup{colorlinks = true, citecolor = MidnightBlue, linkcolor = MidnightBlue, urlcolor = MidnightBlue}

\titleformat{\section}{\normalfont\large\bfseries}{\thesection}{1em}{#1}
\titleformat{\subsection}{\normalfont\normalsize\bfseries}{\thesubsection}{1em}{#1}
\titleformat{\subsubsection}{\normalfont\normalsize\itshape}{\thesubsubsection}{1em}{#1}
\titlespacing\section{0pt}{12pt plus 4pt minus 2pt}{6pt plus 2pt minus 2pt}
\titlespacing\subsection{0pt}{12pt plus 4pt minus 2pt}{3pt plus 2pt minus 3pt}
\titlespacing\subsubsection{0pt}{12pt plus 4pt minus 2pt}{0pt plus 2pt minus 3pt}

\def\boxit#1{\vbox{\hrule\hbox{\vrule\kern6pt
          \vbox{\kern6pt#1\kern6pt}\kern6pt\vrule}\hrule}}

\definecolor{orange}{rgb}{1,0.5,0}
\definecolor{MyDarkBlue}{rgb}{0,0.08,0.45}

\definecolor{MyDarkGreen}{RGB}{0,100,0}

\newtheorem{lemma}{Lemma}[section]

\def\boxit#1{\vbox{\hrule\hbox{\vrule\kern6pt
          \vbox{\kern6pt#1\kern6pt}\kern6pt\vrule}\hrule}}

\definecolor{orange}{rgb}{1,0.5,0}
\definecolor{MyDarkBlue}{rgb}{0,0.08,0.45}

\begin{document}
\title{\Large \bfseries Parameter Estimation for Partially Observed Stable Continuous-State Branching Processes\thanks{We thank DGAPA-PAPIIT-UNAM grant IN105726 for partial financial support.}
 } 
 \author[a]{Eduardo Gutiérrez-Peña\thanks{Corresponding author. \vspace{0.2em} \newline
{\mbox{\hspace{0.47cm}} \it Email addresses:} 
{eduardo@sigma.iimas.unam.mx} (Eduardo Gutiérrez-Peña), \href{mailto:carlos.pmendoza3@gmail.com},
(Carlos Pérez Octavio Mendoza),
\href{mailto:eduardo@sigma.iimas.unam.mx}{carlos.octavio.perez92@gmail.com} , 
\href{mailto:alan@sigma.iimas.unam.mx}{alan@sigma.iimas.unam.mx} (Alan Riva Palacio),
\href{mailto:arno@sigma.iimas.unam.mx}{arno@sigma.iimas.unam.mx} (Arno Siri-Jégousse)
. }}
\author[b]{Carlos Octavio Pérez-Mendoza}
\author[a]{Alan Riva Palacio}
\author[a]{Arno Siri-Jégousse}

\affil[a]{{\small Instituto de Investigaciones en Matemáticas Aplicadas y en Sistemas, UNAM, Mexico City, Mexico}}
\affil[b]{{\small Concordia University, Department of Mathematics and Statistics, Montr\'eal, Canada}}

\vspace{-10pt}
\date{ %\bigskip %\bigskip\bigskip\bigskip \bigskip 
\today}

%%%%%%%%%%%%%%%%%%%%%%%%%%%%%%%%%%%%%%%%%%%%%%%%%%%%%%%%%%%%%%%%%%
%%%%%%%%%%%%%%%%%%%%%%%%%%%%%%%%%%%%%%%%%%%%%%%%%%%%%%%%%%%%%%%%%%
%%% ABSTRACT PAGE
%%%%%%%%%%%%%%%%%%%%%%%%%%%%%%%%%%%%%%%%%%%%%%%%%%%%%%%%%%%%%%%%%%
%%%%%%%%%%%%%%%%%%%%%%%%%%%%%%%%%%%%%%%%%%%%%%%%%%%%%%%%%%%%%%%%%%

\maketitle %\thispagestyle{empty} %\vfill \pagebreak \vspace*{\fill}

%\begin{center} \Large \bfseries Local hedging of variable annuities in the \\ presence of basis risk  \end{center}
%
%\bigskip \bigskip \bigskip \bigskip

%\vspace{-15pt}

\begin{abstract}
\vspace{-9pt}
	
In this article, we present a novel inference framework for estimating the parameters of Continuous-State Branching Processes (CSBPs). We do so by leveraging their subordinator representation. %introduced by \cite{bertoin}. 
Our method reformulates the estimation problem by shifting the stochastic dynamics to the associated subordinator, enabling a parametric estimation procedure without requiring additional assumptions. This reformulation allows for efficient numerical recovery of the likelihood function via Laplace transform inversion, even in models where closed-form transition densities are unavailable. In addition to offering a flexible approach to parameter estimation, we propose a dynamic simulation framework that generates discrete-time trajectories of CSBPs using the same subordinator-based structure.

\bigskip %\bigskip %\bigskip

%\noindent \textbf{JEL classification:} E43, G12.

%\vspace{0.08in}

\noindent \textbf{Keywords:}  Continuous-State Branching Process, Parameter estimation, Subordinator representation.
\end{abstract}

\medskip

% \thispagestyle{empty} \vfill \pagebreak

% \setcounter{page}{1}
% \pagenumbering{roman}

\doublespacing

\setcounter{page}{1}
\pagenumbering{arabic}

\newpage

%%%%%%%%%%%%%%%%%%%%%%%%%%%%%%%%%%%%%%%%%%%%%%%%%%%%%%%%%%%%%%%%%%%%%%%%%
%%%%%%%%%%%%%%%%%%%%%%%%%%%%%%%%%%%%%%%%%%%%%%%%%%%%%%%%%%%%%%%%%%%%%%%%%
%%%%%%%%%%%%%%%%%%%%%%%%%%%%%%%%%%%%%%%%%%%%%%%%%%%%%%%%%%%%%%%%%%%%%%%%%
%%%%%%%%%%%%%%%%%%%%%%%%%%%%%%%%%%%%%%%%%%%%%%%%%%%%%%%%%%%%%%%%%%%%%%%%%

\section{Introduction}\label{se:intro}

Continuous-state branching processes (CSBPs) have been the subject of intense study since their introduction by \cite{Lam67} and \cite{Gre74}, among other authors.
They provide, together with their modified versions, classical models in biology and finance, in particular their continuous path (diffusive) version such as
the Feller diffusion in evolutionary biology and the Cox–Ingersoll–Ross (CIR) process in finance; see \cite{Lam08} or \cite{Li20} for surveys.
These processes provide a flexible framework to capture growth, extinction, and variability in evolving systems governed by randomness.

The use of Feller diffusion in biology is motivated by the branching property. Indeed, it is often assumed that a population evolves as the sum of i.i.d.\ subpopulations.
In finance, the CIR diffusion is heavily used, in part because it provides explicit solutions to important pricing problems.
The only difference between those two models is a linear term that can be easily added to Feller diffusion considering immigration in population's evolution (\cite{Li20}).

For some time, the branching nature of the CIR diffusion, much as its possible connection with Lévy processes (\cite{bertoin, CPU}), was largely unexplored.
The advantages of this property were recently highlighted by \cite{LM10, JMS17, Li20}, in particular when it became necessary to consider modifications of the CIR diffusion that include jumps.
This is the purpose of the $\alpha$-CIR process, which provides a parametric model for interest rates based on the $\alpha$-stable version of CSBPs.
In \cite{JMS17}, this model was used for bond pricing and a maximum likelihood estimator of the parameters was obtained and studied asymptotically in the case where the time between observations tends to 0.

In biology, CSBPs and their generalizations are commonly used to model the evolution of the total size of a population for certain phenomena, see e.g.\ \cite{Lam05, MPS17, GCPP21}. Most statistical applications have been developed for their discrete version, being modifications of the Galton-Watson process, as in \cite{Jan76, DY94, Extremadura13, Extremadura21}. 
For continuous-state and continuous-time processes, the problem of estimating model parameters from discrete-time observations remains challenging, particularly when the underlying process involves jumps or heavy tails, as in the case of $\alpha$-stable branching mechanisms. Estimation becomes even more involved when closed-form density expressions are unavailable.

%Our work: estimation of parameters from time-series of observations, consistency while the number of observations increases. Method based on inversion of Laplace transform. 

%{\bfseries WORK ON THE LITERATURE REVIEW}

%Branching processes play a central role in the modeling of stochastic population dynamics across disciplines, from biology to finance. In particular, the Continuous-State Branching Process (CSBP) family has found extensive applications in fields such as genetics, epidemiology, and interest rate modeling. These processes provide a flexible framework to capture growth, extinction, and variability in evolving systems governed by randomness.

%There is a rich body of literature on the modeling and analytical properties of CSBPs, as well as on their applications to pricing, risk management, and population forecasting. However, the problem of estimating model parameters from discrete-time observations remains challenging, particularly when the underlying process involves jumps or heavy tails, as in the case of $\alpha$-stable branching mechanisms. Estimation becomes even more delicate when closed-form density expressions are unavailable.

In this work, we develop a likelihood-based parameter estimation framework for CSBPs observed at discrete time intervals. Our approach is based on a probabilistic representation of the process via its associated subordinator, which enables us to recover the likelihood function using numerical inversion of the Laplace transform. This method accommodates a wide range of branching mechanisms and does not require explicit forms for the transition densities. We show that the procedure yields consistent parameter estimates as the number of observations increases, offering a principled route for parameter estimation in these models.

%Parameter estimation in our setting incorporates both likelihood-based point estimation and density approximation techniques. 
In addition to maximum likelihood estimation, we propose a sampling-based Bayesian method that provides Monte Carlo samples from the posterior distribution of the parameters. These samples can be used to produce point and interval estimates, or any other form of inference.
%we complement our approach with density estimates obtained through importance sampling within a Bayesian framework, offering a flexible alternative for uncertainty quantification. 
To validate our methodology, we present a series of numerical experiments based on simulated data. We also describe a simulation scheme for generating discrete-time sample paths of CSBPs when only the branching mechanism is known, a tool that is useful both for testing and for potential applications in synthetic data generation.

The remainder of the paper is organized as follows. In Section \ref{se:the_model}, we introduce the class of CSBPs under consideration. Section \ref{se:inversion} presents the numerical inversion of the Laplace transform, along with the simulation method and likelihood construction. Section \ref{se:parameter_estimation} develops the parameter estimation strategy, with a detailed discussion of the $\alpha$-stable case. In Section \ref{se:experiments}, we illustrate our procedure using simulated datasets. Finally, Section \ref{se:conclusion} offers concluding remarks and outlines directions for future work.

%%%%%%%%%%%%%%%%%%%%%%%%%%%%%%%%%%%%%%%%%%%%%%%%%%%%%%%%%%%%%%%%%%%%%%%%%
%%%%%%%%%%%%%%%%%%%%%%%%%%%%%%%%%%%%%%%%%%%%%%%%%%%%%%%%%%%%%%%%%%%%%%%%%
%%%%%%%%%%%%%%%%%%%%%%%%%%%%%%%%%%%%%%%%%%%%%%%%%%%%%%%%%%%%%%%%%%%%%%%%%
%%%%%%%%%%%%%%%%%%%%%%%%%%%%%%%%%%%%%%%%%%%%%%%%%%%%%%%%%%%%%%%%%%%%%%%%%
\section{The model}\label{se:the_model}

A continuous-state branching process (CSBP) is a $\mathbb R_+$-valued Markov process $X=(X_{t},t\geq 0)$ that fulfills the branching property: its Laplace transform satisfies, for any $x,y\geq0$,

 $$
    \mathbb{E}\left[\exp(-\lambda X_{t}) \mid X_{0}=x+y \right] =     \mathbb{E}\left[\exp(-\lambda X_{t}) \mid X_{0}=x \right] \mathbb{E}\left[\exp(-\lambda X_{t}) \mid X_{0}=y \right].
    $$
In words, the process starting at state $x+y$ is the sum of two independent copies of the process started at $x$ and $y$.
We refer to \cite{Kyp14} for a complete survey.
The branching property entails that its Laplace transform has the form
\begin{equation*}
    \mathbb{E}\left[\exp(-\lambda X_{t}) \mid X_{0}=x \right] = \exp\{ -xu_{t}(\lambda)\},
\end{equation*}

where $u_{t}$ is the solution of the equation
\begin{equation}\label{eq:ut}
    \frac{\partial}{\partial t}u_{t}(\lambda)=-\Psi (u_{t}(\lambda))
\end{equation}
and $\Psi:[0,\infty) \rightarrow \mathbb{R}$ is a convex function with the Lévy-Khintchine representation
\begin{equation}
    \Psi (u)=\gamma u+\frac{\sigma^2}{2} u^{2}+\int_{\mathbb{R}_{+}}(e^{ -xu} -1+xu\mathbb{I}_{\{x\leq 1\}})\Pi(dx)
\end{equation}
for $\gamma \in \mathbb{R}$, $\sigma \geq 0$ and $\Pi$ a measure on $\mathbb{R}_{+}$ such that $\int (1\wedge x^2)\Pi(dx)<\infty$.
The measure $\Pi$ encodes the jumps of the stochastic process. Indeed  CSBPs admit a stochastic differential representation given by

\begin{equation*}
    dX_{t} = \gamma X_{t}dt + \sigma\sqrt{X_t}dW_t+\int_{(0,t)}\int_{(0,\infty)}\int_{(0,X_{s-})} x\tilde{N}(ds,dx,dv)
\end{equation*}
with $W$ a standard Brownian motion and $\tilde{N}(ds,dx,dv)$ a Poisson point process on $[0,\infty)^{3}$ with intensity $ds\otimes \Pi(dx)\otimes dv$. From this representation, we can see that CSBPs generalize in a very natural way the CIR process, whose SDE is
\begin{equation*}
    dX_{t} = a(b- X_{t})dt + \sigma\sqrt{X_t}dW_t,
\end{equation*}
by adding a jump part to the classical diffusion.
The missing linear part is easily added by considering CSBPs with immigration, see \cite{LM10, JMS17}. 
We will focus on the case $b=0$ since we are mainly interested in the statistical inference encoded by the jump part of the process.

We are interested in stable CSBPs obtained when $\sigma=0$ (no diffusive part) and $\Pi(dx)=c(\alpha,\beta)x^{-1-\alpha}dx$ {for a well chosen constant $c(\alpha,\beta)$} depending on $\alpha\in(1,2)$ and $\beta>0$.
In this case, the branching mechanism simplifies to $$\Psi(u)=\beta u^{\alpha}+\gamma u.$$
Note that the trajectories are given by (small and accumulating) jumps since $\int (1\wedge x)\Pi(dx)=\infty$. The jumps get very small and the trajectories converge to those of a diffusion process as $\alpha\to2$. The jumps become larger when $\alpha$ gets closer to 1.

By plugging $\Psi$ into \eqref{eq:ut}, we can see that $u_{t}$ is a solution to the equation

\begin{equation*}
    \frac{\partial}{\partial t} u_{t}(\lambda) = - \beta u_{t}(\lambda)^{\alpha} - \gamma u_{t}(\lambda).
\end{equation*}
It is thus easy to check that
{\begin{equation}\label{eq:utsol}
u_{t}(\lambda) =  e^{ -\gamma t }\left( \lambda^{1-\alpha} + \frac{\beta}{\gamma}(1-e^{\gamma(1-\alpha)t } )\right)^{\frac{1}{1-\alpha}}.
\end{equation}}
This is an explicit formulation for the Laplace transform of $X_t$ for any $t\geq0$; see Example 3.1 in  \cite{Li12}.

CSBPs are called subcritical (resp.\ supercritical) if $\Psi'(0+)<0$ (resp.\ $>0$).
Subcritical processes become extinct (they reach 0 in finite time) with probability 1.
In our case of study, criticality is determined by the sign of 
$\gamma$.
In any case, the probability of being already extinct at a certain time $t$ is
{\begin{equation}\label{eq:probability}
   p_t(0)= \mathbb{P}\left( X_{t}=0 \mid X_{0}=x \right)=e^{-xu_t(\infty)}=
\exp\left\{-xe^{ -\gamma t }\left(  \frac{\gamma}{\beta(1-e^{\gamma(1-\alpha)t } )}\right)^{\frac{1}{\alpha-1}}\right\}.
\end{equation}}
This is important for the calibration of parameters, as we will usually try to control this quantity when observing positive realizations of $X_t$. Furthermore, we note that 
$$
\mathbb{E}\left[ e^{-\lambda X_t} \,\big|\, X_0=x\right] = p_t(0) + (1-p_t(0))\mathbb{E}\left[ e^{-\lambda X_t} | X_t \neq 0,\, X_0=x\right],
$$
so we can obtain the Laplace transform of the process conditioned on not being zero.

There exists a remarkable connection between CSBPs and subordinators \cite{bertoin}.
Conditional on $X_0=x$, the random variable $X_t$ has the same law as $S_x$ where the process $S=(S_x,x\geq0)$ is a subordinator with Lévy exponent $u_t$. That is

\begin{equation}\label{eq:lapt}
\mathbb{E}\left[ e^{-\lambda S_x} \right]=e^{-xu_t(\lambda)}.    
\end{equation}
This correspondence is convenient for simulation purposes. Thanks to the Markov property, we will be able to obtain discretized trajectories of a CSBP $(X_{\delta}, X_{2\delta}, \dots, X_{k\delta})$ by simulating independent realizations of the associated subordinator at different times $(S_{X_0}, S_{X_\delta}, \dots, S_{X_{(k-1)\delta}})$ for a time-step $\delta\geq0$.
The independent and stationary increments of Lévy processes make them easier objects to simulate, see e.g. \cite{App09}. 

%%%%%%%%%%%%%%%%%%%%%%%%%%%%%%%%%%%%%%%%%%%%%%%%%%%%%%%%%%%%%%%%%%%%%%%%%
%%%%%%%%%%%%%%%%%%%%%%%%%%%%%%%%%%%%%%%%%%%%%%%%%%%%%%%%%%%%%%%%%%%%%%%%%
%%%%%%%%%%%%%%%%%%%%%%%%%%%%%%%%%%%%%%%%%%%%%%%%%%%%%%%%%%%%%%%%%%%%%%%%%
%%%%%%%%%%%%%%%%%%%%%%%%%%%%%%%%%%%%%%%%%%%%%%%%%%%%%%%%%%%%%%%%%%%%%%%%%

\section{Laplace transform inversion and simulation framework}\label{se:inversion}

This section introduces the methodology used to retrieve the density function of the subordinator $(S_{t},{t \geq 0})$ at any given time $t$. We also describe the simulation scheme employed in our numerical experiments.

\subsection{Numerical inversion of the Laplace transform}\label{se:laplace_inversion}

To recover the probability density function of the subordinator $(S_{t}, {t \geq 0})$, we rely on the representation introduced by \cite{bertoin}. In particular, conditional on $X_{(k-1)\delta} = x$, the distribution of the subordinator $(S_{t}, {t \geq 0})$ associated with $X_{k\delta}$ is characterized at time $t = x$ via its Laplace transform as
\begin{equation*}
    \mathbb{E}\left[ e^{-\lambda S_{x}} \right] = e^{-xu_{\delta}(\lambda)}.
\end{equation*}

This representation allows us to approximate the density function of $S_t$ via numerical inversion of its Laplace transform. This approach is particularly convenient when the subordinator arises from an $\alpha$-stable jump process, as the corresponding probability density function does not admit a closed-form expression. In such cases, numerical inversion offers an efficient and flexible way to obtain accurate density estimates without relying on computationally expensive simulation techniques.

To carry out this inversion, we adopt the algorithm proposed by \cite{abate2000introduction}, which is well-suited for situations where the Laplace transform is known but its inverse is not analytically tractable. Let $f^{S}_t(s)$ and $F^{S}_t(s)$ denote the density and distribution function of the subordinator at time $t$, respectively. These are recovered from the Laplace transform $\mathcal{L}\{f^{S}_{t}\}(\lambda) = e^{-t\,u_\delta(\lambda)}$ using a numerical quadrature approach based on the trapezoidal rule with step size $h$ applied to the Bromwich integral
\begin{align*}
f_t^S(s)& = \frac{1}{2\pi i} \int_{r - i \infty}^{r + i \infty}
e^{ \lambda s - t u_{\delta}(\lambda) }\mathrm{d}\lambda,
\end{align*}
where $r>0$ must be greater than the real part of all the singularities of $u_\delta(\lambda)$.

The numerical inversion method of \cite{abate2000introduction} is controlled by four parameters: integers $\ell$, $m$, $n$ and $A>0$, which govern the trade-off between computational cost and accuracy. For evaluation of the inverse at value $s$, the quantity $A=2 r \ell s $ determines the aliasing error, which decays exponentially with larger values, while $h=\pi / \ell s$ is the step size of the trapezoidal rule and $n,m$ determine a Euler summation scheme for truncation of the series in the quadrature scheme. As established in \cite{ridout2009generating}, the error from using the trapezoidal rule to approximate the Bromwich integral is bounded by $e^{-A}/(1-e^{-A})$, which is approximately $10^{-8}$ for $A = 18.4$. This framework has been successfully applied to $\alpha$-stable distributions, yielding accurate results. In our work, we have to take special care with the choice of $A$ as the Laplace transforms we consider present singularities with positive real parts for our choice of $1<\alpha<2$. 

\begin{lemma}\label{lema3.1}
The singularities corresponding to the Laplace transform  \eqref{eq:lapt} with $u_t$ in \eqref{eq:utsol} are given by
$$
\left\{  z_n= \left( \frac{\gamma}{\beta\left(1-e^{-\gamma(\alpha-1)t}\right)} \right)^{\frac{1}{\alpha-1}}e^{ \frac{(\pi+2\pi n)}{\alpha-1}i } \,:\, n\in \mathbb{Z}\right\}
$$
\end{lemma}

For the remaining parameters, we chose values $\ell = 1$, $m = 11$, and $n = 38$, which provide a good balance between accuracy and computational efficiency. Although exact error bounds are not shown here, the reliability of the method when closed-form densities are unavailable can be empirically verified by comparing with theoretical results such as the analytical formulas for the mean of $f_t^S$, limit values at zero and infinity provided by the initial and final value theorems for Laplace transforms, and the corresponding rates of convergence provided by Tauberian theorems. We present such results and the proof of Lemma \autoref{lema3.1} in \autoref{appendix}.

\subsection{Simulation approach for CSBP dynamics}\label{subsec:simulation_path}

The simulation of a discrete path of the CSBP is performed through independent realizations of its associated subordinator at different time steps. Given an initial value $X_0$, a simulated path of the CSBP with step size $\delta$ is given by the mapping
\begin{equation}
X_{i\delta} = S_{X_{(i-1)\delta}},
\end{equation}
where $X_{(i-1)\delta}$ is the value of the CSBP at the previous time step. This means that, once $X_{(i-1)\delta}$ is known, the value of $X_{i\delta}$ can be simulated as the subordinator $S_t$ evaluated at $t = X_{(i-1)\delta}$.

The simulation of the subordinator $S_t$ is performed by drawing $S_t=0$ with probability $p_t(0)$ given by \eqref{eq:probability} and using the Laplace transform inversion described in the previous Section \ref{se:laplace_inversion} for
\begin{equation}\label{eq:laplace_transforme}
    \mathbb{E}\left[ e^{-\lambda X_t} \,\big|\, X_t \neq 0,\, X_0=x\right] = \frac{e^{-xu_t (\lambda)}-p_t(0)}{1-p_t(0)}.
\end{equation}
Our method consists in using inverse transform sampling by solving the root of 
\begin{equation}
F^{S}_{t}(s) = \mathbb{P}\left[ S_t \leq s \right] = u.
\end{equation}
where $F^{S}_{t}(s)$ is obtained numerically using the \cite{abate2000introduction} method, and $u \sim U(0,1)$ is a uniform random variate. This approach follows the RLAPTRANS algorithm introduced by \cite{ridout2009generating}. As shown in \cite{ridout2009generating}, the inversion method provides a reliable and efficient source of random numbers, with minimal user input and good performance even for complex distributions such as $\alpha$-stable laws. Our method follows the same principles and offers a similarly robust mechanism for simulating subordinators.

%%%%%%%%%%%%%%%%%%%%%%%%%%%%%%%%%%%%%%%%%%%%%%%%%%%%%%%%%%%%%%%%%%%%%%%%%
%%%%%%%%%%%%%%%%%%%%%%%%%%%%%%%%%%%%%%%%%%%%%%%%%%%%%%%%%%%%%%%%%%%%%%%%%
%%%%%%%%%%%%%%%%%%%%%%%%%%%%%%%%%%%%%%%%%%%%%%%%%%%%%%%%%%%%%%%%%%%%%%%%%
%%%%%%%%%%%%%%%%%%%%%%%%%%%%%%%%%%%%%%%%%%%%%%%%%%%%%%%%%%%%%%%%%%%%%%%%%

\section{Parameter estimation procedure}\label{se:parameter_estimation}

In this section, we describe the methodology used to estimate the parameters of a CSBP based on a discrete-time sample. We begin by outlining a general likelihood-based framework that leverages the Markovian structure of the process and the Laplace inversion techniques discussed previously. Then, we present a tailored two-step estimation strategy designed to address the lack of identifiability specific to the $\alpha$-stable CSBP model studied in this work.

\subsection{Likelihood-based estimation via subordinator densities}

Our parameter estimation procedure is based on recovering the likelihood function associated with the observed sample of the process. Leveraging the Markov property of the CSBP, the joint likelihood can be factorized into a product of conditional densities. These conditional densities correspond to the law of the subordinators associated with the process and can be estimated using the Laplace inversion techniques described in Section~\ref{se:laplace_inversion}.

Suppose we are given a discrete-time observation of the CSBP, $(X_{0},X_{\delta}, X_{2\delta}, \dots, X_{k\delta})$. By the Markov property, the joint density of the process can be expressed as
\begin{equation}
    f(X_{\delta}, X_{2\delta}, \dots, X_{k\delta}\mid X_{0}\,;\theta) = \prod_{i=1}^{k} f(X_{i\delta} \mid X_{(i-1)\delta}\,; \theta),
\end{equation}
where $\theta$ denotes the vector of parameters of the CSBP. For each $i = 1, \dots, k$, the Laplace transform of the conditional density $f(X_{i\delta} \mid X_{(i-1)\delta}; \theta)$ is given by
\begin{equation*}
    \mathbb{E}\left[\exp(-\lambda X_{i\delta}) \mid X_{(i-1)\delta} = x \right] = e^{-x\,u_\delta(\lambda)},
\end{equation*}
where $u_\delta(\lambda)$ is the Laplace exponent associated with the CSBP over time interval $\delta$. These expressions coincide with those obtained for the subordinators $(S_{X_0}, S_{X_\delta}, \dots, S_{X_{(k-1)\delta}})$, and thus the conditional densities can be recovered by applying the same numerical inversion methods discussed previously.

Consequently, the likelihood function for estimating $\theta$ is approximated by
\begin{equation}
    \mathcal{L}(\theta \mid X_{0},X_{\delta}, X_{2\delta}, \dots, X_{k\delta}) \approx \prod_{i=1}^{k} f^{S}_{X_{(i-1)\delta}}(X_{i\delta}\,; \theta),
\end{equation}
where $f^{S}_{X_{(i-1)\delta}}(X_{i\delta}\,; \theta)$ denotes the numerically recovered density, evaluated at $X_{i\delta}$, of the marginal law of the subordinator at time $X_{(i-1)\delta}$. The parameter vector $\theta$ is then estimated by maximizing the approximated likelihood function using standard maximum likelihood techniques. Formally, the estimator $\theta^{*}$ is defined as
\begin{equation*}
    \theta^{*} = \underset{\theta \in \Theta}{\arg\max} \; \mathcal{L}(\theta \mid X_{0},X_{\delta}, X_{2\delta}, \dots, X_{k\delta}),
\end{equation*}
where $\Theta$ denotes the parameter space.

\subsection{Two-step optimization for the $\alpha$-stable CSBP}

While the previous subsection outlines a general and flexible framework for parameter estimation in CSBPs, its direct application to specific models may present additional challenges. In the case of the $\alpha$-stable CSBP considered in this work, maximization of the likelihood function is complicated by identifiability issues among the model parameters. As a result, a two-step estimation procedure is required to ensure convergence toward stable and interpretable parameter values.\footnote{An example illustrating this issue in practice is provided in section~\ref{subsub:identifiability}, where we show how direct likelihood maximization over $(\gamma, \beta, \alpha)$ fails to recover the true parameters due to lack of identifiability.} These challenges are inherently model-dependent, and for other families of CSBPs, the complexity of the estimation task may vary depending on the analytical properties of the branching mechanism and the behavior of the associated subordinator.

The parameters to be estimated in the $\alpha$-stable CSBP model are $\gamma$, $\beta$, and $\alpha$, which together form the parameter vector $\theta = (\gamma, \beta, \alpha) \in \mathbb{R} \times \mathbb{R}_{+} \times (1,2)$. The parameter $\alpha$ governs the frequency and magnitude of the jumps, and therefore plays a central role in shaping the tail behavior of the marginal distributions. In contrast, the parameters $\gamma$ and $\beta$ determine the long-term dynamics of the process, influencing whether the population tends toward extinction or unbounded growth.

Because of the dominant role of $\alpha$ in determining the overall distributional structure, and due to the identifiability interactions among parameters, we adopt a two-step estimation strategy. This approach is reminiscent of the two-stage estimation procedure often applied to Student-$t$ distributions, where the degrees of freedom, which control tail heaviness, are estimated separately after obtaining preliminary estimates for the location and scale parameters. In our case, we first estimate $(\gamma, \beta)$ conditionally on $\alpha$, and subsequently optimize over $\alpha$ using the resulting likelihood surface. This sequential strategy improves numerical stability and leads to more reliable convergence of the estimation process.

With this in mind, our estimation procedure approximates the likelihood surface by discretizing the domain of $\alpha$ and solving the conditional optimization problem at each point of the resulting partition $P = \langle \alpha_{1}, \dots, \alpha_{m} \rangle$, for $n \in \mathbb{N}$. For a fixed value $\alpha_j$, we solve

\begin{equation}\label{eq:estimation_gamma_beta}
(\gamma_{j}, \beta_{j}) = \underset{(\gamma, \beta) \in \mathbb{R} \times \mathbb{R}_{+}}{\arg\max}  \mathcal{L}((\gamma, \beta, ; \alpha_{j}) \mid X_{0},X_{\delta}, X_{2\delta}, \dots, X_{k\delta}),
\end{equation}

and then identify the parameter triplet $(\gamma_{J}, \beta_{J}, \alpha_{J})$ that maximizes the overall likelihood, where the index $J$ is given by

\begin{equation}\label{eq:full_estimation}
J = \underset{j \in \mathbb{N}}{\arg\max} ; \{ \mathcal{L}((\gamma_{j}, \beta_{j}, \alpha_{j}) \mid X_{0}, X_{\delta}, X_{2\delta}, \dots, X_{k\delta}) \}.
\end{equation}

This approach effectively decouples the optimization over $\alpha$ from that of the remaining parameters, thereby addressing identifiability issues while preserving computational tractability.

%%%%%%%%%%%%%%%%%%%%%%%%%%%%%%%%%%%%%%%%%%%%%%%%%%%%%%%%%%%%%%%%%%%%%%%%%
%%%%%%%%%%%%%%%%%%%%%%%%%%%%%%%%%%%%%%%%%%%%%%%%%%%%%%%%%%%%%%%%%%%%%%%%%
%%%%%%%%%%%%%%%%%%%%%%%%%%%%%%%%%%%%%%%%%%%%%%%%%%%%%%%%%%%%%%%%%%%%%%%%%
%%%%%%%%%%%%%%%%%%%%%%%%%%%%%%%%%%%%%%%%%%%%%%%%%%%%%%%%%%%%%%%%%%%%%%%%%

\section{Simulation-based evaluation}\label{se:experiments}

In this section, we present numerical experiments designed to evaluate the performance of our proposed method for parameter estimation in the CBSP model. These experiments aim to demonstrate both the effectiveness of the approach in accurately recovering model parameters and its robustness with respect to the choice of step size and other tuning parameters.

\subsection{Simulated environment}\label{subsec:simulation_env}

We conducted simulations across a range of controlled scenarios to evaluate the stability of the estimation procedure and to examine its behavior under varying configurations of the underlying stochastic dynamics. Although our method does not require multiple trajectories to estimate the parameters of the $\alpha$-CSBP with parameters $(\gamma, \beta, \alpha)$, we assess its empirical performance using multiple simulated paths. This allows us to characterize the distribution and robustness of the parameter estimates across a wide variety of settings.

The simulated environment is defined by two key hyperparameters: the time-step size $\delta$ and the number of time steps $n$. To assess the sensitivity of our method to the discretization level, we consider two values for the time-step size, $\delta = \frac{1}{6}$ and $\delta = \frac{1}{12}$.\footnote{In practical applications, the time-step size of the observed time series can often be rescaled. This setup is therefore intended as a general illustration of the method’s sensitivity to this hyperparameter.} The number of time steps is set to $n = 10$ and $n = 20$, allowing us to evaluate the method’s robustness with respect to the sample size and the effective time horizon, given by $t(n, \delta) = n \delta$. This dual variation enables an assessment of our method across different temporal resolutions and observation lengths for the simulated $\alpha$-CSBP process.

Each experiment is based on $m = 100$ simulated trajectories. The values of $(\gamma, \beta, \alpha)$ are selected to ensure that the probability of extinction remains low. This design choice guarantees that trajectories have observations at all time steps, which is essential for our illustrations, particularly when assessing the sensitivity of the estimation procedure to the time series length through the total duration $t(n, \delta)$.

\subsection{Numerical experiments}\label{subsec:numerical_results}

\subsubsection{Lack of identifiability}\label{subsub:identifiability}

The $\alpha$-CSBP model, parameterized by $(\gamma, \beta, \alpha)$, presents challenges related to lack of identifiability, a situation in which different sets of parameter values yield the same or nearly indistinguishable probability distributions for the observed data. In such cases, it becomes difficult to uniquely recover the true parameters based solely on the likelihood, thereby compromising the reliability of maximum likelihood estimation.

To illustrate this issue, consider the $\alpha$-CSBP model with parameters $(\gamma, \beta, \alpha) = (-6, 6, 1.5)$, using a time-step size of $\delta = 1/6$ and a total of $n = 20$ time steps. This setting serves as a counterexample to highlight how the direct application of the likelihood maximization procedure can lead to inaccurate parameter recovery.

Specifically, after simulating $m = 100$ trajectories under this parameter configuration, we attempt to estimate all three parameters simultaneously for each path in order to assess the quality and stability of the estimation procedure:

\begin{equation*}
(\gamma^{*}, \beta^{*}, \alpha^{*}) = \underset{(\gamma, \beta, \alpha) \in \mathbb{R} \times \mathbb{R}^{+} \times (1,2)}{\arg\max} \ \mathcal{L}((\gamma, \beta, \alpha) \mid X_{\delta}, X_{2\delta}, \dots, X_{k\delta}).
\end{equation*}

Despite being estimated from data simulated under the correct model, the resulting parameter estimates exhibit substantial deviations from their true values, particularly for $\gamma$ and $\alpha$, which govern the extinction probability and jump frequency components, respectively. The average estimates for these parameters are $-1.79$ for $\gamma$ and $1.25$ for $\alpha$. These discrepancies signal an underlying identifiability issue, as illustrated in \autoref{fig:density}. The empirical densities of the estimates display clear multimodal features. These patterns suggest that multiple parameter combinations produce equally plausible fits to the data, undermining the reliability of standard maximum likelihood estimation. This behavior underscores the need for a more structured and robust estimation procedure, such as the two-step strategy proposed in this work.

\begin{figure}[H]\centering
\caption{Empirical density plots of the estimated model parameters.}
\includegraphics[width=17cm]{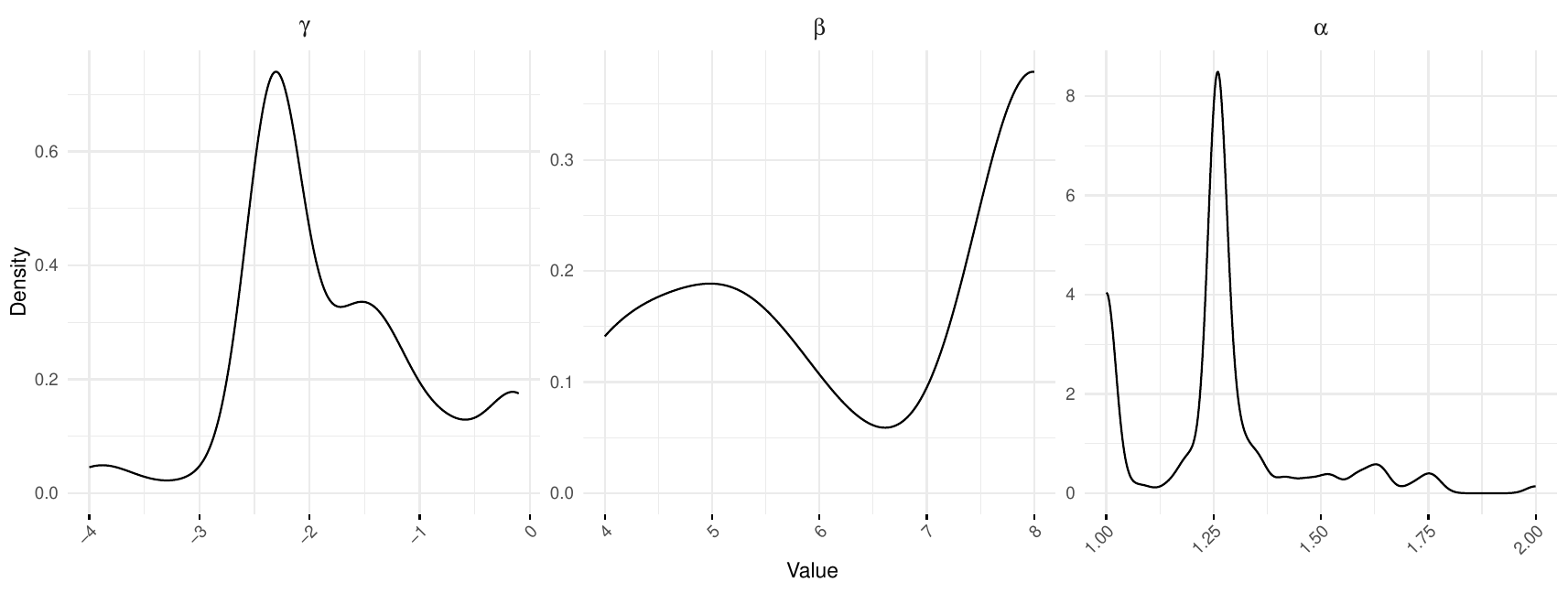}
\begin{tablenotes}
\item Parameter estimates are obtained using the maximum likelihood method and averaged over 100 independently simulated paths.
\end{tablenotes}
\label{fig:density}
\end{figure}

\subsubsection{Stability of parameter estimates across time frames and sample sizes}\label{subsub:stability_time_size}

We continue our analysis by simulating sample paths of the $\alpha$-CSBP with the parameter vector $(\gamma, \beta, \alpha) = (-6, 1/\delta, 1.5)$. This specific configuration is chosen for illustrative purposes, as it guarantees that the probability of extinction, given by Equation~\eqref{eq:probability}, remains extremely low over the two time horizons considered: $t = \frac{10}{12},\, \frac{10}{6},\, \frac{20}{12} \text{ and } \frac{20}{6}$. In all cases, the extinction probability is bounded above by $1.3631\times10^{-6}$, ensuring that the process remains active throughout the observation window and that the trajectories are informative for estimation.

The complete experimental procedure is described in Algorithm~\ref{algo:experiment_procedure}. The algorithm outlines the steps to simulate multiple paths under the known parameter configuration and, treating the parameters as unknown, estimate them independently for each simulated trajectory. 

\begin{algorithm}[H]
\caption{Simulation and estimation procedure for the $\alpha$-CSBP model.}\label{algo:experiment_procedure}
\begin{algorithmic}
\State \textbf{Input 1:} Set parameter values $(\gamma, \beta, \alpha)$.
\State \textbf{Input 2:} Define a grid of candidate $\alpha$ values for the estimation process: $\alpha \in \{1.1, 1.2, \dots, 1.9\}$.

\State Assuming 
\For{$k=1$ to $m$}
    \Comment{Procedure per-path basis}
    \State Simulate a trajectory $\left\{X^{(k)}_{i\delta}\right\}_{i=0}^{n}$ using the simulation scheme described in Section~\ref{subsec:simulation_path}.
    \For{$\alpha \in \{1.1, 1.2, \dots, 1.9\}$}
        \State Estimate values $\left(\gamma_{\alpha}^{(k)},\beta^{k}_{\alpha}\right)$ as described in Equation \eqref{eq:estimation_gamma_beta}.
    \EndFor
    \State Estimate optimal values $\left(\gamma^{(k)*}, \beta^{(k)*}, \alpha^{(k)*}\right)$ as described in Equation \eqref{eq:full_estimation}.
\EndFor
\end{algorithmic}
\end{algorithm}

Our numerical results, shown in \autoref{fig:alpha_estimates}, display the average log-likelihood values computed over a grid of $\alpha$ values, based on multiple simulated paths and following the procedure outlined in Algorithm~\ref{algo:experiment_procedure}. The results show that the value of $\alpha$ that maximizes the likelihood is consistently $\alpha = 1.5$, regardless of the combination of time-step size and sample size used in the simulation. This demonstrates that our method reliably recovers the true jump parameter $\alpha$ and that its performance is not affected by the total duration of the path, $t(n, \delta)$.

\begin{figure}[H]\centering
\caption{Mean log-likelihood over the grid of $\alpha$ values.}
\includegraphics[width=17cm]{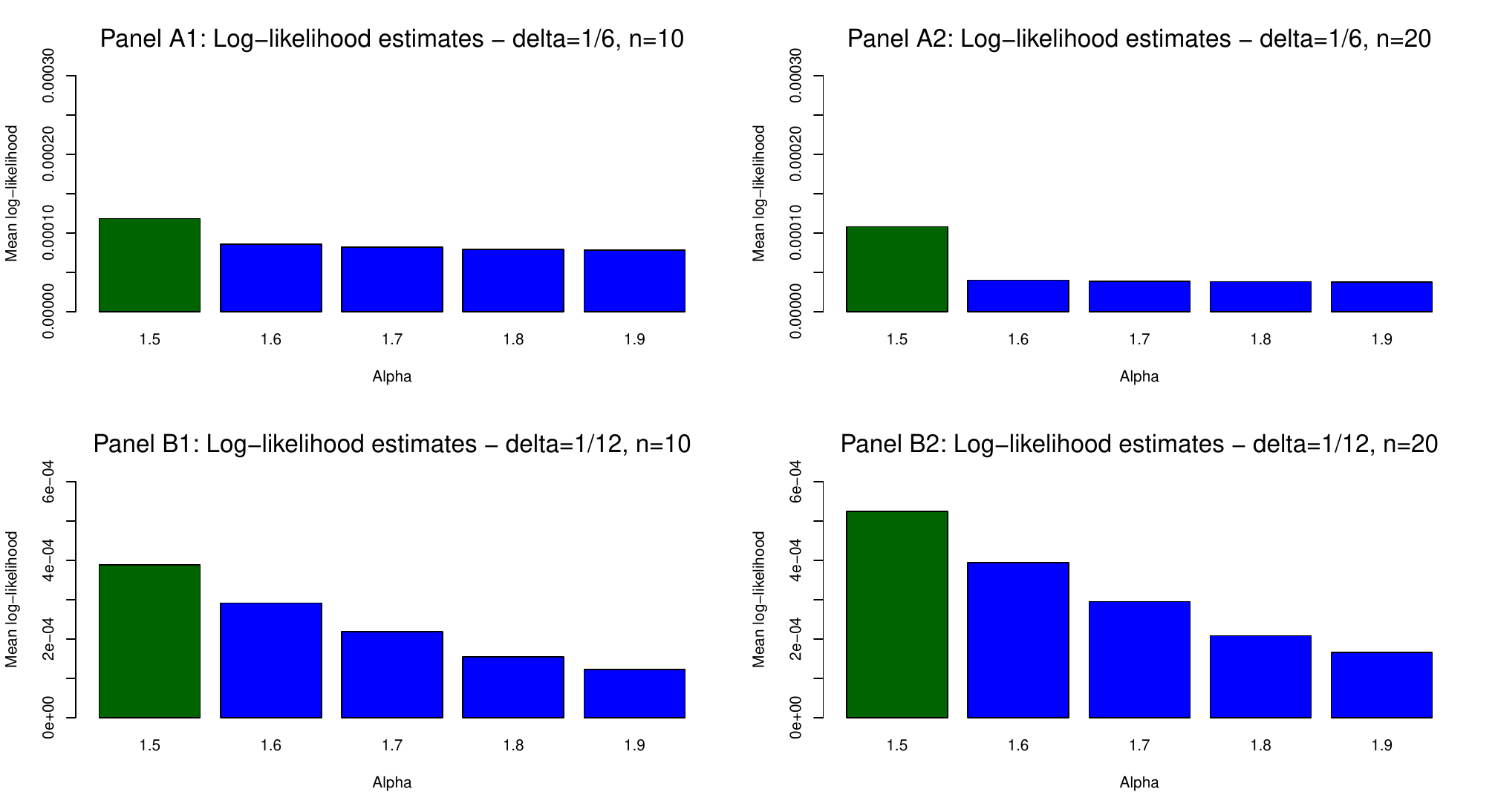}
\begin{tablenotes}
\item The results are obtained by averaging the likelihood values computed using Equation~\eqref{eq:full_estimation} across 100 simulated paths. The estimates of the parameters $(\gamma, \beta)$ depend on the value of $\alpha$ and are not shown in this figure.
\end{tablenotes}
\label{fig:alpha_estimates}
\end{figure}

As expected, the likelihood increases as the time-step size decreases and the sample size increases, reflecting reduced uncertainty in the observed path due to the availability of more information. This reduction in uncertainty with respect to the time-step size is evident in panels B1 and B2 which exhibit higher likelihood values than their counterparts, Panels A1 and A2, as a result of the time-step size reduction from $1/6$ to $1/12$. Similarly, the effect of sample size is observed in the comparison between left and right panels: Panels A1 and B1 (with $n = 10$) display lower likelihoods than Panels A2 and B2 (with $n = 20$), highlighting the impact of increased sample size on the accuracy of the likelihood estimates.

Regarding the estimates of the parameters $\gamma$ and $\beta$, \autoref{tab:beta_gamma_estimates} presents the mean estimates across all simulated paths along with their corresponding 95\% confidence intervals. The results demonstrate the robustness of our method, as reflected by consistently narrow confidence intervals for both parameters across all configurations of the hyperparameters $\delta$ and $n$.

\begin{table}[H]
\centering
\footnotesize
\caption{Mean and 95\% confidence intervals for parameters $\gamma$ and $\beta$.}
\label{tab:beta_gamma_estimates}
\renewcommand{\arraystretch}{1.3}
\resizebox{\textwidth}{!}{
\begin{tabular}{>{\centering\arraybackslash}p{3cm} >{\raggedright\arraybackslash}p{5cm} >{\raggedright\arraybackslash}p{5cm}}
\toprule
$\delta$ / $n$ & $n = 10$ & $n = 20$ \\
\midrule
\multirow{2}{*}{$\delta = \frac{1}{6}$}
& $\gamma = -5.908\ [-6.034,\ -5.781]$ & $\gamma = -5.969\ [-6.011,\ -5.926]$ \\
& $\beta = 6.031\ [5.988,\ 6.073]$ & $\beta = 6.019\ [5.988,\ 6.051]$ \\
\midrule
\multirow{2}{*}{$\delta = \frac{1}{12}$}
& $\gamma = -5.889\ [-6.015,\ -5.763]$ & $\gamma = -5.954\ [-6.044,\ -5.864]$ \\
& $\beta = 12.051\ [11.950,\ 12.152]$ & $\beta = 12.031\ [11.969,\ 12.094]$ \\
\bottomrule
\end{tabular}}
\begin{tablenotes}
\item 
Mean and 95\% confidence intervals for the estimated parameters $(\gamma, \beta)$ using the estimated $\alpha$ values from \autoref{fig:alpha_estimates}. Results are based on $m = 100$ simulated paths.
\end{tablenotes}
\end{table}

Importantly, the confidence intervals for $\gamma$ and $\beta$ consistently contain the true value used in the simulation, providing strong evidence of accurate recovery. Overall, these findings highlight the method's capacity to produce stable and accurate parameter estimates across the given range of time-step sizes and sample sizes.

\subsubsection{Method robustness to variations in the jump parameter $\alpha$}

As a second experiment, we identify the regions of $\alpha$ under which our method maintains consistent performance. More precisely, we determine the values of $\alpha$ for which the numerical inversion adopted in \cite{abate2000introduction} can be reliably implemented.

The numerical inversion of the Laplace transform proposed by \cite{abate2000introduction} is mainly driven by the parameter $A$, which controls the aliasing error. As described in \autoref{se:laplace_inversion}, $A = 2 r \ell s$ with $\ell = 1$, $s$ denotes the current value of the stochastic process, and $r$ is any value located to the right of all singularities of the Laplace transform specified in  \autoref{eq:laplace_transforme}. Based on Lemma \autoref{lema3.1}, such an $r$ can be chosen as
$$r = \left( \frac{\gamma}{\beta\left(1-e^{-(\alpha-1) \gamma \delta}\right)} \right)^\frac{1}{\alpha-1},$$
which provides an upper bound for the real part of all singularities.

For illustration purposes, let us consider the trivial case in which the current value of the stochastic process is equal to 1. In this setting, the contribution of the process value to the parameter $A$ is negligible, and therefore $A = 2r$. If we set $\gamma = -1/\delta$, the resulting value of $A$ remains largely insensitive to variations in $\beta$ and $\delta$, which, following the numerical example of \autoref{subsub:stability_time_size}, are fixed at $\beta = 6$ and $\delta=1/6$. Under this configuration, the parameter $A$ can be viewed effectively as a function of $\alpha$ only.

Our numerical results reveal that as $\alpha$ approaches 1, the value of $A$ required to guarantee the validity of the numerical Laplace inversion increases dramatically. \autoref{fig:A_param} illustrates the behavior of $A$ as a function of $\alpha$ under several configurations of the parameter $\beta$.

\begin{figure}[H]\centering
\caption{Growth of $A$ as a function of $\alpha$.}
\includegraphics[width=17cm]{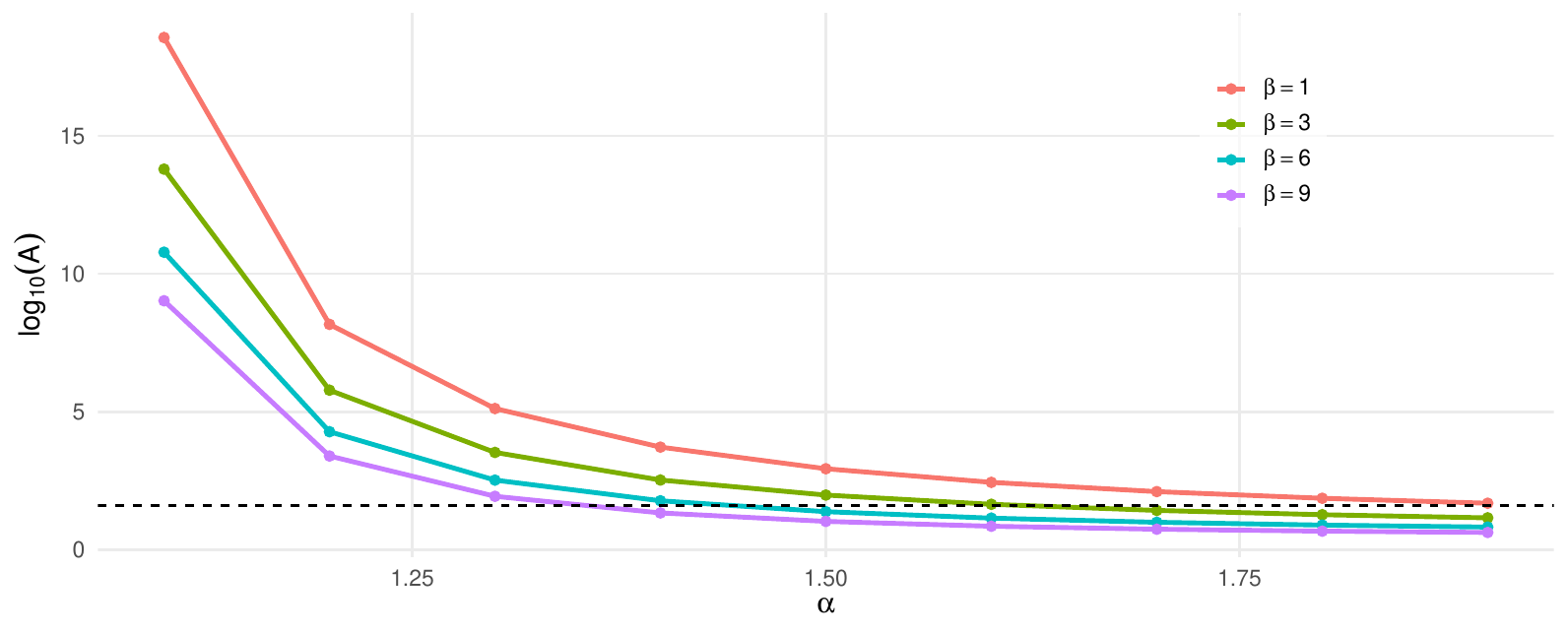}
\begin{tablenotes}
\item Values of the parameter $A$ as a function of $\alpha$ for different $\beta$ values. The sharp increase near $\alpha = 1$ indicates reduced numerical stability. The dashed line marks the region of stable robust results.
\end{tablenotes}
\label{fig:A_param}
\end{figure}

To ensure stability within the dynamic framework, where the Laplace transform must be inverted at each time step for different values of $s = x$, we consider experiments where $\alpha \geq 1.5$. This choice has consistently produced stable and robust behavior in our simulation scheme. When $\alpha$ is smaller, both the simulation and the inference procedures become susceptible to numerical instability, which may introduce bias into the final results. An alternative approach to controlling the magnitude of the parameter $A$ is to increase the value of $\beta$; however, doing so exponentially increases the probability of extinction, thereby reducing the number of observable data points over time and limiting the amount of informative data available for estimation.

Conversely, as $\alpha$ approaches 2, the numerical inversion becomes substantially more stable, and neither the simulations nor the inference procedure exhibit numerical issues. For illustration, we focus on the representative case $\alpha = 1.9$, using a time-step size of $\delta = \frac{1}{6}$ and a total of $n = 20$ time steps. While other values of $\delta$ and $n$ produce similar behavior, this configuration succinctly demonstrates the robustness of our approach when $\alpha$ is near~2. The parameter values $(\gamma, \beta)$ are set to $(-6, 6)$.
This choice of parameters, intended solely to preserve non-zero observations along the sample paths, does not limit the applicability of the estimation method. In practice, estimation remains feasible even when paths become extinct early, by using the available observations up to the extinction time. This simply results in fewer data points per trajectory, a scenario under which our method has already demonstrated reliable performance in Section~\ref{subsub:stability_time_size}.

The empirical results shown in Panel A of \autoref{fig:alpha_estimates_diff_alpha} display the average log-likelihood across all sample paths obtained by applying Algorithm~\ref{algo:experiment_procedure} under two scenarios when trajectories are simulated with $\alpha = 1.9$. The results indicate that the maximum average log-likelihood is attained at the true value of $\alpha$, successfully recovering the original parameter used in the simulation. This highlights the effectiveness of our method in identifying the correct $\alpha$ value, even when it lies near the right boundary of its domain.

Regarding the estimates of the parameters $(\gamma, \beta)$, Panel B of \autoref{tab:beta_gamma_estimates_alphas} presents the average values obtained under the estimated $\alpha$. These estimates are based on 100 simulated trajectories, enabling the construction of confidence intervals and a robust assessment of the method’s reliability and stability. The results demonstrate that our methodology consistently yields narrow confidence intervals for $(\gamma, \beta)$, successfully recovering values that are very close to those used in the simulation. This highlights the robustness of the proposed approach.

\begin{figure}[H]
\centering
\caption{Statistical results of the estimation procedure under simulations with $\alpha=1.9$}
\begin{minipage}{0.47\textwidth}
    \centering
    \caption*{{\footnotesize Panel A: Average log-likelihood over the grid of $\alpha$ values.}}
    \includegraphics[width=\textwidth]{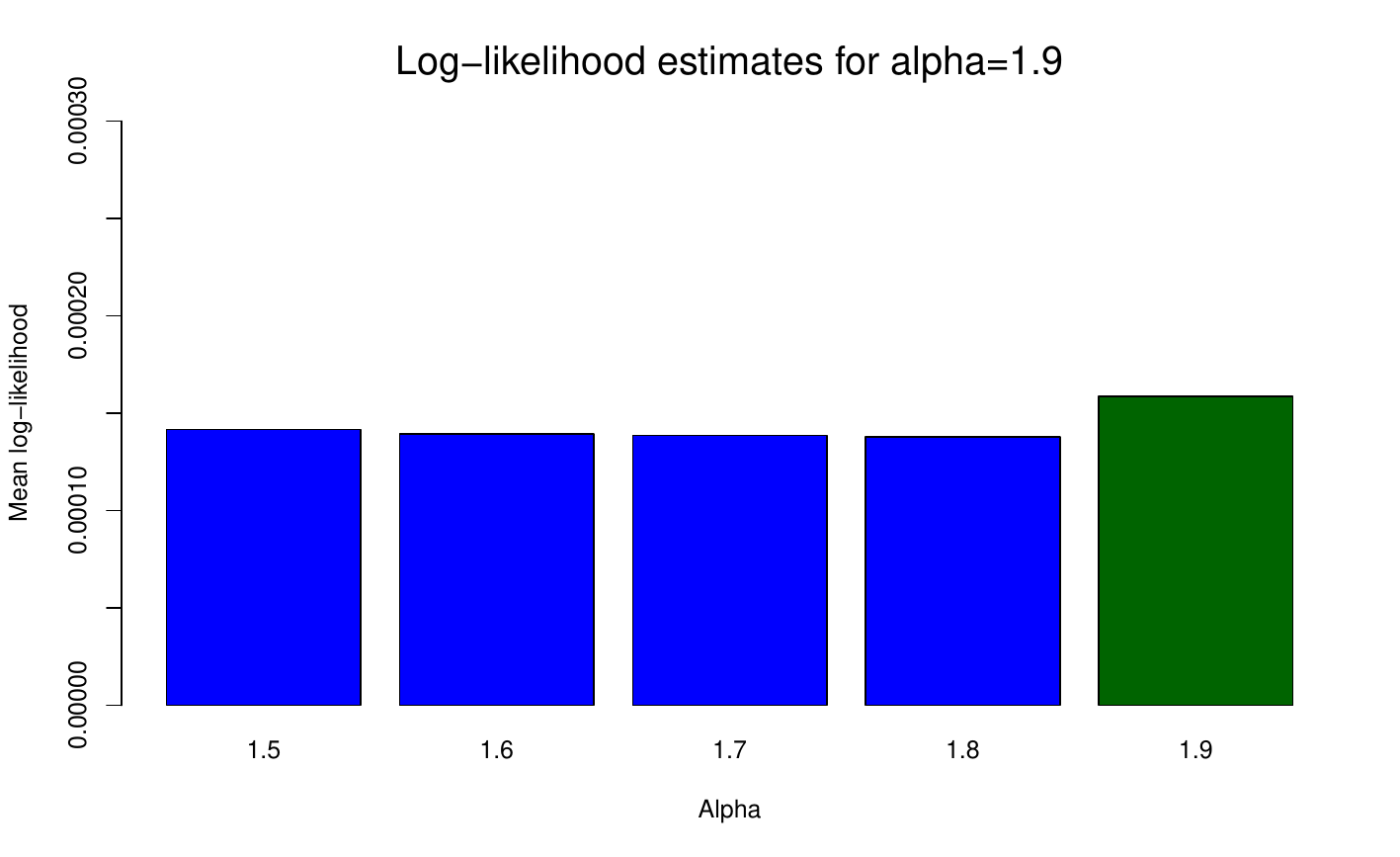}
    \begin{tablenotes}{\footnotesize
        \item The results are obtained by averaging the likelihood values computed using Equation~\eqref{eq:full_estimation} across 100 simulated paths. The estimates of the parameters $(\gamma, \beta)$ depend on the value of $\alpha$ and are not shown in this figure.}
    \end{tablenotes}
    \label{fig:alpha_estimates_diff_alpha}
\end{minipage}
\hfill
\begin{minipage}{0.47\textwidth}
    \centering
    \footnotesize
    \caption*{{\footnotesize Panel B: Mean and 95\% confidence intervals for parameters $\gamma$ and $\beta$.}}
    \label{tab:beta_gamma_estimates_alphas}
    \renewcommand{\arraystretch}{1.3}
    \resizebox{\textwidth}{!}{
    \begin{tabular}{>{\centering\arraybackslash}p{3cm} >{\raggedright\arraybackslash}p{5cm}}
        \toprule
        Time-steps & $\boldsymbol{\alpha=1.9}$ \\
        \midrule
        \multirow{2}{*}{$n = 50$}
        & $\gamma = -6.119\ [-6.254,\ -5.984]$ \\
        & $\beta = 6.067\ [5.856,\ 6.278]$ \\
        \bottomrule
    \end{tabular}}
    \begin{tablenotes}{\footnotesize
        \item Mean and 95\% confidence intervals for the estimated parameters $(\gamma, \beta)$ using the estimated $\alpha$ values from \autoref{fig:alpha_estimates_diff_alpha}. Results are based on $m = 100$ simulated paths for extreme values of $\alpha = 1.9$.}
    \end{tablenotes}
\end{minipage}
\end{figure}

\subsubsection{A Bayesian estimation framework via importance sampling}\label{subsub:bayesian}

In this section, we explore a Bayesian approach to parameter estimation using importance sampling. This methodology is introduced as a tool to estimate the distribution of the parameters $(\gamma, \beta, \alpha)$, offering an alternative to frequentist methods. The Bayesian framework enables us to incorporate prior beliefs and to quantify uncertainty in a principled manner, resulting in a full posterior distribution over the parameters of interest.

The general idea of importance sampling is to approximate the posterior distribution $p(\gamma, \beta, \alpha \mid \mathcal{D})$ by reweighting samples drawn from a proposal distribution $q(\gamma, \beta, \alpha)$. In the basic form considered here, we use the prior as the proposal, i.e., $q = p_{\text{prior}}$. This approach is appealing both for its simplicity and for its interpretability: the prior serves a dual role by representing our initial beliefs about the parameters and by providing a sampling mechanism for approximating the posterior. The posterior distribution is given by Bayes’ rule as
\begin{equation}
p(\gamma, \beta, \alpha \mid \mathcal{D}) \propto \mathcal{L}(\mathcal{D} \mid \gamma, \beta, \alpha) p(\gamma, \beta, \alpha),
\end{equation}
where $\mathcal{L}(\mathcal{D} \mid \gamma, \beta, \alpha)$ denotes the likelihood of the data given the parameters, and the prior density $p(\gamma, \beta, \alpha)$ encode any available domain knowledge or modeling assumptions. We assume that the priors for $\gamma$, $\beta$, and $\alpha$ are independent, so the joint prior factorizes as $p(\gamma, \beta, \alpha) = p(\gamma)p(\beta)p(\alpha)$.

Given $N$ i.i.d.\ samples $\{(\gamma^{(i)}, \beta^{(i)}, \alpha^{(i)})\}_{i=1}^N \sim p_{\text{prior}}$, we compute importance weights proportional to the likelihood:
\begin{equation}
w^{(i)} = \mathcal{L}(\mathcal{D} \mid \gamma^{(i)}, \beta^{(i)}, \alpha^{(i)}), \quad \tilde{w}^{(i)} = \frac{w^{(i)}}{\sum_{j=1}^N w^{(j)}}.
\end{equation}
These normalized weights define a discrete approximation of the posterior distribution. Posterior expectations of any function $g(\beta, \gamma, \alpha)$ are then approximated by a weighted average:
\begin{equation}
\mathbb{E}_{\text{post}}[g(\gamma, \beta, \alpha)] \approx \sum_{i=1}^N \tilde{w}^{(i)} g(\gamma^{(i)}, \beta^{(i)}, \alpha^{(i)}).
\end{equation}
This strategy enables full posterior inference while transparently incorporating prior information in both the modeling and computational steps.

As an illustrative example, we consider the case of an $\alpha$-CSBP model with true parameters $(\gamma, \beta, \alpha) = (-6, 6, 1.5)$. Assuming access to prior knowledge or plausible domain-informed information, we specify independent prior distributions for each parameter, which also serve as the proposal distribution in the importance sampling procedure. The prior for $\alpha$, constrained to the interval $(1, 2)$, is modeled using a transformed Beta distribution:
\begin{equation}\label{eq:alpha}
\alpha \sim 1 + \mathrm{Beta}(3, 3),
\end{equation}
which centers its mass around $1.5$. The prior for $\gamma \in \mathbb{R}$ is taken as a normal distribution:
\begin{equation}\label{eq:gamma}
\gamma \sim \mathcal{N}(-6, 1^2),
\end{equation}
reflecting strong prior belief near $-6$. Finally, the prior for $\beta > 0$ is assumed to be log-normal:
\begin{equation}\label{eq:beta}
\beta \sim \mathrm{Lognormal}(1.8, 0.1^2),
\end{equation}
which concentrates its mass near $6$ while preserving positive support.

We start our experiment by simulating $N = 1{,}000$ samples of the vector $(\gamma^{(i)}, \beta^{(i)}, \alpha^{(i)})$ from equations \eqref{eq:alpha}, \eqref{eq:gamma}, and \eqref{eq:beta}. For illustration purposes, each sample is then evaluated in the likelihood function derived from a simulated trajectory with step size $\delta = 1/24$ and $n = 50$ time steps. This allows us to compute the corresponding importance weights and obtain a weighted empirical approximation of the posterior distribution of the parameters $(\gamma, \beta, \alpha)$. This approach enables us to quantify uncertainty through probability intervals and variance estimates, and to assess how the likelihood updates our initial beliefs encoded in the prior.

\autoref{fig:importance_sampling} shows the approximate posterior density functions of the three parameters. As expected, the sample updates the prior beliefs based on the observed trajectory. It is important to note that a well-chosen initial estimate through the prior density can help alleviate the identifiability issues discussed in Section~\ref{subsub:identifiability}, allowing for accurate parameter estimation based on a single trajectory. In fact, the estimates obtained from our methodology in Section~\ref{subsub:stability_time_size} could serve as a suitable point of mass concentration for the prior distribution, facilitating posterior inference and enabling uncertainty quantification through probability intervals or posterior variances.

\begin{figure}[H]\centering
\caption{Posterior distributions.}
\includegraphics[width=17cm]{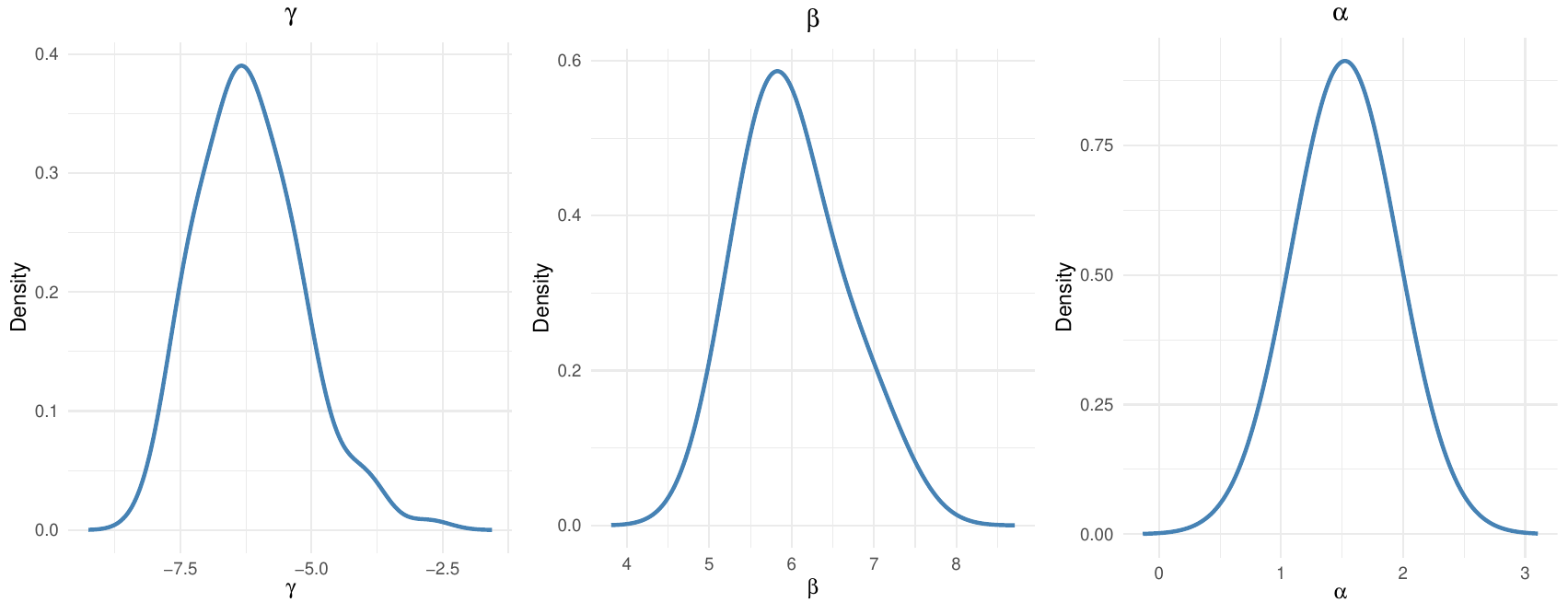}
\begin{tablenotes}
\item The results are based on $N = 1{,}000$ samples of the vector $(\gamma^{(i)}, \beta^{(i)}, \alpha^{(i)})$ generated from equations \eqref{eq:alpha}, \eqref{eq:gamma}, and \eqref{eq:beta}. The posterior distribution is estimated via importance sampling, using weights proportional to the likelihood function evaluated on a single simulated trajectory with step size $\delta = 1/24$ and $n = 50$ time steps.
\end{tablenotes}
\label{fig:importance_sampling}
\end{figure}

As an illustration, Table~\ref{tab:posterior_estimates} presents the expected value and corresponding variance of each parameter based on the posterior distribution. %This summary provides a compact view of the central tendency and uncertainty associated with each parameter estimate.

\begin{table}[H]
\centering
\footnotesize
\caption{Expected values and variances of the parameters based on the posterior distribution.}
\label{tab:posterior_estimates}
\renewcommand{\arraystretch}{1.3}
\resizebox{\textwidth}{!}{
\begin{tabular}{>{\centering\arraybackslash}p{3cm} >{\raggedright\arraybackslash}p{5cm} >{\raggedright\arraybackslash}p{5cm}}
\toprule
Parameter & Mean & Variance \\
\midrule
$\gamma$
& $-6.177\ $ & $(0.1740)^2$ \\
$\beta$
& $\phantom{-}6.013\ $ & $(0.572)^2$ \\
$\alpha$
& $\phantom{-}1.522\ $ & $(0.174)^2$ \\
\bottomrule
\end{tabular}}
\begin{tablenotes}
\item 
Results are based on $N = 1{,}000$ samples of $(\gamma^{(i)}, \beta^{(i)}, \alpha^{(i)})$ from equations \eqref{eq:alpha}, \eqref{eq:gamma}, and \eqref{eq:beta}. The posterior distribution is estimated using importance sampling with weights proportional to the likelihood of a single simulated trajectory with step size $\delta = 1/6$ and $n = 20$ time steps.
\end{tablenotes}
\end{table}

%%%%%%%%%%%%%%%%%%%%%%%%%%%%%%%%%%%%%%%%%%%%%%%%%%%%%%%%%%%%%%%%%%%%%%%%%
%%%%%%%%%%%%%%%%%%%%%%%%%%%%%%%%%%%%%%%%%%%%%%%%%%%%%%%%%%%%%%%%%%%%%%%%%
%%%%%%%%%%%%%%%%%%%%%%%%%%%%%%%%%%%%%%%%%%%%%%%%%%%%%%%%%%%%%%%%%%%%%%%%%
%%%%%%%%%%%%%%%%%%%%%%%%%%%%%%%%%%%%%%%%%%%%%%%%%%%%%%%%%%%%%%%%%%%%%%%%%

\section{Concluding remarks}\label{se:conclusion}

In this work, we developed a novel inference framework for parameter estimation in continuous-state branching processes (CSBPs), leveraging their representation via subordinators. By shifting the inference problem to the associated subordinator and using numerical inversion of the Laplace transform, we proposed a flexible and efficient likelihood-based method that remains applicable even when closed-form transition densities are unavailable.

Our approach accommodates $\alpha$-stable CSBPs, which are particularly relevant in applications involving heavy-tailed or jump-driven dynamics. To address the identifiability challenges inherent in these models, we introduced a two-step estimation procedure that separates the estimation of the jump parameter $\alpha$ from the other model parameters. This strategy leads to more stable and accurate recovery of the full parameter vector $(\gamma, \beta, \alpha)$. The ability to numerically compute the likelihood function, also allowed us to produce Bayesian inferences by means of a simple importance sampling scheme to approximate the posterior distribution of the parameters.
%In addition to point estimation via maximum likelihood, the framework can be complemented with Bayesian density estimation, for example by incorporating importance sampling approximations of the likelihood when posterior inference is desired.

While our inference framework shows strong performance across most of the parameter domain, special care is required when $\alpha$ approaches 1. In this regime, the numerical inversion parameters increase rapidly, leading to potential instability in both simulation and likelihood evaluation. Although this issue can be mitigated by selecting larger values of $\beta$, thereby increasing the extinction probability, such an adjustment is often undesirable, as it reduces the amount of observable data and may limit the practical applicability of the framework. Nevertheless, our results demonstrate that the method delivers consistent and accurate estimates over a substantial portion of the $\alpha$ domain, thereby extending existing statistical inference techniques to a broader class of $\alpha$-stable CSBPs under both traditional likelihood-based and Bayesian approaches. An additional contribution of this work is the development of a reliable and efficient simulator for $\alpha$-stable CSBPs, which may serve as a valuable tool for future methodological and applied research in this family of branching processes.

Through simulation studies, we demonstrated the robustness and effectiveness of the proposed method. The results show that the estimator reliably recovers the true parameter values across a range of scenarios, including varying time-step sizes, sample sizes, and values of $\alpha$ near the boundaries of its domain. In all cases, the method produced narrow confidence intervals and consistent likelihood surfaces, even under sparse or finely discretized observations.

Overall, this framework provides a principled and scalable solution for inference in a class of branching processes that are central to both biological modeling and financial applications. Future work may explore extensions to fully Bayesian implementations, or address more general classes of CSBPs beyond the $\alpha$-stable case.

%%%%%%%%%%%%%%%%%%%%%%
\vspace{-10pt}
\titleformat{\section}{\normalfont\large\bfseries} {\thesection}{1em}{#1}
\bibliographystyle{apalike}
%\bibliographystyle{abbrvnat}
%\bibliographystyle{unsrtnat}
% plain
\bibliography{references}  

%\pagebreak
%\subfile{Online appendix}

%\newpage

\appendix

\section{Theoretical results for CSBP subordinators}\label{appendix}

We denote the Laplace transform 

$$
L(\lambda) =  \frac{e^{-xu_t (\lambda)}-p_t(0)}{1-p_t(0)}
$$
with $u_t(\lambda)$ given by \eqref{eq:utsol} and $p_t(0)$ given by \eqref{eq:probability} in the main text. We observe that $L$ is completely monotone.
%that density exists and is ultimately monotone). 

\begin{lemma}
The singularities corresponding to the Laplace transform  \eqref{eq:lapt} with $u_t$ in \eqref{eq:utsol} are given by
$$
\left\{  z_n= \left( \frac{\gamma}{\beta\left(1-e^{-\gamma(\alpha-1)t}\right)} \right)^{\frac{1}{\alpha-1}}e^{ \frac{(\pi+2\pi n)}{\alpha-1}i } \,:\, n\in \mathbb{Z}\right\}
$$
which has finite cardinality when $\alpha-1$ is rational and infinite when it is not.
\end{lemma}
\begin{proof}
We observe that when evaluating in the complex plane the singularities of $L$ coincide with the pole singularities of 
\begin{align*}
u_{t}(\lambda) &= \frac{ e^{ -\gamma t }} { \left( \frac{1}{\lambda^{\alpha - 1}} + \frac{\beta}{\gamma}(1-e^{-\gamma(\alpha-1)t } )\right)^{\frac{1}{\alpha-1}}}
\end{align*}
which become essential singularities of the form $\exp(1/z)$ at $z=0$ for $L$. Such singularities satisfy
$$
\frac{1}{\lambda^{\alpha - 1}} + \frac{\beta}{\gamma}(1-e^{-\gamma(\alpha-1)t } ) = 0
$$
so they are given by the set
$$
\left\{  z_n= \left( \frac{\gamma}{\beta\left(1-e^{-\gamma(\alpha-1)t}\right)} \right)^{\frac{1}{\alpha-1}}e^{ \frac{(\pi+2\pi n)}{\alpha-1}i } \,:\, n\in \mathbb{Z}\right\}
$$
\end{proof}

\begin{lemma}{Limit behavior at zero and infinity.}
$$
f(0^+) := \lim_{t\to 0^+ }f(t) = \infty
$$
\end{lemma}
\begin{proof}
By the initial value theorem for Laplace transforms
\begin{align*}
f(0^+) & = \lim_{\lambda \to \infty } \lambda L_t(\lambda ) =
\lim_{y \to 0 } \frac{ L_t\left( \frac{1}{y}  \right) }{y} 
\\ & = \frac{1}{1- p_t(0) } \lim_{y \to 0 } \frac{  e^{-x u_t\left(\frac{1}{y}\right)} - e^{-xu_t(\infty)} }{y}
\\ & = \frac{-x p_t(0) e^{-\gamma t}}{1-p_t(0) } \frac{\partial}{\partial y}\left( \frac{1}{ \left( y^{\alpha - 1} + \frac{\beta}{\gamma}(1-e^{-\gamma(\alpha-1)t } )\right)^{\frac{1}{\alpha-1}} } \right)
\Bigg|_{y=0}
\\ & = \frac{x p_t(0) e^{-\gamma t}}{1-p_t(0) } \left( \frac{y^{\alpha-2} \left( y^{\alpha - 1} + \frac{\beta}{\gamma}(1-e^{-\gamma(\alpha-1)t } )\right)^{\frac{1}{\alpha-1}-1 } }{ \left( y^{\alpha - 1} + \frac{\beta}{\gamma}(1-e^{-\gamma(\alpha-1)t } )\right)^{\frac{2}{\alpha-1} } }  \right)
\Bigg|_{y=0}
\\ & = \infty.
\end{align*}
On the other hand by the final value theorem of Laplace transforms
\begin{align*}
f(\infty) := \lim_{ t \to \infty } f(t) & = \lim_{\lambda \to 0 } \lambda L_t(\lambda ) =
\lim_{\lambda \to 0 } \lambda \left( \frac{e^{-xu_t(\lambda)}-e^{-xu_t(\infty)}}{1-e^{-xu_t(\infty)}} \right)=
\lim_{\lambda \to 0 } \lambda = 0
\end{align*}
\end{proof}

\begin{lemma}{Tail behavior of $X_\delta$.}

\begin{equation*}
    \int_0^\infty \mathbb{P}\left( X_{\delta}>y \mid X_{0}=x \right) dy= xe^{ -\gamma \delta },
\end{equation*}
and, for the process conditioned to survival,
\begin{equation*}
    \int_0^\infty \mathbb{P}\left( X_{\delta}>y \mid X_t\neq 0,X_{0}=x \right) dy= \frac{xe^{ -\gamma \delta }}{1-p_0},
\end{equation*}
and
\begin{equation*}
     \mathbb{P}\left( X_{\delta}\leq y \mid X_t\neq 0,X_{0}=x \right) dy\sim
     \frac{p_0}{1-p_0}\frac{e^{ -\gamma \delta }}{\alpha-1}\left(  \frac{\gamma}{\beta(1-e^{\gamma(1-\alpha)\delta } ) } \right)^{\frac{\alpha}{\alpha-1}}y^{{\alpha-1}}
\end{equation*}
as $y\to0$.
\end{lemma}

\begin{proof}
We will apply the Tauberian theorem, see for example 
Corollary 8.1.7 in \cite{Bingham}. Since
 $u_\delta(\lambda)\sim\lambda e^{ -\gamma \delta }$ as $\lambda\to0$,
we get
$$
      1-\mathbb{E}\left[e^{-\lambda X_{\delta}} \mid X_{0}=x \right]\sim x e^{ -\gamma \delta }\lambda.
$$
On the other hand,  as $\lambda\to\infty$,
$$\mathbb{E}\left[e^{-\lambda X_{\delta}} \mid X_\delta\neq0,X_{0}=x \right]\sim 
\frac{p_0}{1-p_0}\frac{e^{ -\gamma \delta }}{\alpha-1}\left(  \frac{\gamma}{\beta(1-e^{\gamma(1-\alpha)\delta })} \right)^{\frac{\alpha}{\alpha-1}}\lambda^{{1-\alpha}}.$$
It just remains to apply Theorem XIII.5.3  in \cite{Feller}.
\end{proof}

\end{document}